\documentclass[12pt]{article}

\usepackage{amsmath}
\usepackage{amssymb}
\usepackage{amsthm}
\usepackage{fullpage}

\newtheorem{theorem}{Theorem}
\newtheorem{corollary}[theorem]{Corollary}
\newtheorem{lemma}[theorem]{Lemma}

\newcommand{\prob}[2]{\mathop{\mathrm{Pr}}_{#1}[#2]}
\newcommand{\avg}[2]{\mathop{\textbf{E}}_{#1}\left[#2\right]}

\newcommand{\F}{\mathbb{F}}

\newcommand{\mc}[1]{\mathcal{#1}}

\newcommand{\Supp}{\mathrm{Supp}}
\newcommand{\VP}{\mathrm{VP}}

\newcommand{\VNP}{\mathrm{VNP}}

\newcommand{\coeff}{\mathrm{Coeff}}
\newcommand{\mon}{\mathfrak{m}}

\usepackage{tikz}

\title{Strongly Exponential Separation Between Monotone $\VP$ and Monotone $\VNP$}
\author{Srikanth Srinivasan\thanks{Email: \texttt{srikanth@math.iitb.ac.in}}\\ Department of Mathematics\\ IIT Bombay}

\begin{document}

\maketitle

\begin{abstract}
We show that there is a sequence of explicit multilinear polynomials $P_n(x_1,\ldots,x_n)\in \mathbb{R}[x_1,\ldots,x_n]$ with non-negative coefficients that lies in monotone $\VNP$ such that any monotone algebraic circuit for $P_n$ must have size $\exp(\Omega(n)).$ This builds on (and strengthens) a result of Yehudayoff (2018) who showed a lower bound of $\exp(\tilde{\Omega}(\sqrt{n})).$
\end{abstract}

\section{Introduction}

This paper deals with a problem in \emph{Algebraic Complexity}, which is the study of the complexity of computing multivariate polynomials over some underlying field $\F$. The model of computation is the \emph{Algebraic circuit} model, which computes polynomials from $\F[x_1,\ldots,x_n]$ using the basic sum and product operations in this ring. This model and its variants have been studied by a large body of work (see, e.g. the surveys~\cite{ShpilkaYehudayoff, Ramprasadgit}). 

The central question in the area is Valiant's~\cite{Valiant} $\mathrm{VP}$ vs. $\mathrm{VNP}$ question. The set $\mathrm{VP}$ contains sequences $(P_n(x_1,\ldots,x_n))_{n\geq 1}$ of polynomials of polynomially bounded degree\footnote{i.e. $\deg(P_n)\leq n^{O(1)}$} that can be computed by polynomial-sized algebraic circuits. The class $\mathrm{VNP}$ contains sequences $(Q_n(x_1,\ldots,x_n))_{n\geq 1}$ where 
\[
Q_n(x_1,\ldots,x_n) = \sum_{b_1,\ldots,b_m\in \{0,1\}}P_{n+m}(x_1,\ldots,x_n,b_1,\ldots,b_m)
\] 
where $m$ is polynomially bounded in $n$ and $(P_{r}(x_1,\ldots,x_r))_{r\geq 1}$ is in $\mathrm{VP}.$

Like its Boolean analogue, the $\mathrm{VP}$ vs. $\mathrm{VNP}$ question has proved stubbornly hard to resolve, the principal bottleneck being our inability to prove explicit algebraic circuit lower bounds. Given this, it is natural to look at variants of this question. 

In a recent paper~\cite{amir}, Yehudayoff considered the \emph{monotone} version of the $\mathrm{VP}$ vs. $\mathrm{VNP}$ question, which is defined as follows. The underlying field is $\mathbb{R}$ and the polynomials being computed have non-negative coefficients. A \emph{monotone} algebraic circuit is one where all the constants appearing in the circuit are non-negative. The monotone versions of $\mathrm{VP}$ and $\mathrm{VNP}$, denoted $\mathrm{MVP}$ and $\mathrm{MVNP}$ respectively, are defined analogously: $\mathrm{MVP}$ contains (sequences of) polynomials that have small \emph{monotone} algebraic circuits; $\mathrm{MVNP}$ contains (sequences of) polynomials that can be written as exponential Boolean sums over polynomials in $\mathrm{MVP}.$

Monotone algebraic circuits have been studied since the 80s, and explicit exponential lower bounds are known for this model via the work of Schnorr~\cite{Schnorr} and Jerrum and Snir~\cite{JerrumSnir} (see also~\cite{Valiantnegation,ShamirSnir,gashkov-sergeev,RazYehudayoff}). However, as Yehudayoff~\cite{amir} pointed out, these results do not imply a separation between $\mathrm{MVP}$ and $\mathrm{MVNP}$. In fact, most\footnote{The one exception to this seems to be a lower bound of Raz and Yehudayoff~\cite{RazYehudayoff}. Here, it is unclear whether the hard polynomials lie in $\mathrm{MVNP}$ but we are unable to rule it out.} of the monotone circuit lower bounds proved in earlier work also imply that the same polynomials do not belong to $\mathrm{MVNP},$ and hence do not imply a separation between these two classes. 

The main result of~\cite{amir} was the resolution of the $\mathrm{MVP}$ vs. $\mathrm{MVNP}$ question. More precisely, Yehudayoff showed that there is an explicit sequence of multilinear polynomials $(P_n(x_1,\ldots,x_n))_{n\geq 1}$ in $\mathrm{MVNP}$ such that any monotone algebraic circuit for $P_n$ must have size $\exp(\tilde{\Omega}(\sqrt{n})).$

In this paper, we strengthen this result to a strongly exponential lower bound.

\begin{theorem}
\label{thm:mainintro}
There is an explicit sequence of multilinear polynomials $(P_n(x_1,\ldots,x_n))_{n\geq 1}$ in $\mathrm{MVNP}$ such that any monotone algebraic circuit for $P_n$ must have size $2^{\Omega(n)}.$
\end{theorem}

This theorem bears a similar relation to Yehudayoff's result as some later works~\cite{gashkov-sergeev,RazYehudayoff} bears to the result of Schnorr~\cite{Schnorr}. Schnorr~\cite{Schnorr} proved a lower bound of $\exp(\Omega(\sqrt{n}))$ for an explicit family of polynomials; a similar lower bound was also proved for an explicit family of polynomials by Jerrum and Snir~\cite{JerrumSnir}.\footnote{These explicit polynomials were based on the Clique and the Permanent respectively.} These bounds were strengthened to strongly exponential lower bounds by a series of works of Kuznetsov, Kasim-Zade, and Gashkov in the USSR in the 80s~\cite{Kasim-Zade,Gashkov,gashkov-sergeev}\footnote{Unfortunately, journal versions of these papers are not easily available, but we refer to a survey of Gashkov and Sergeev~\cite{gashkov-sergeev} for a very interesting account of this line of work, along with details of some of these results.}, and independently by a more recent result of Raz and Yehudayoff~\cite{RazYehudayoff}.

\subsection{Proof Outline}

\paragraph{High level idea.} We rely on a connection between monotone algebraic circuit lower bounds and communication complexity that was made explicit by Raz and Yehudayoff~\cite{RazYehudayoff}. As shown in~\cite{RazYehudayoff}, if a multilinear polynomial $P\in \mathbb{R}[x_1,\ldots,x_n]$ has a monotone algebraic circuit of size $s$, then we get a decomposition
\begin{equation}
\label{eq:prod-decomp}
P = \sum_{i=1}^{s} g_i h_i
\end{equation}
where each summand $g_ih_i$ satisfies the property that $g_i$ and $h_i$ are \emph{non-negative} multilinear polynomials that depend on disjoint sets of at least $n/3$ variables each. We call each such term a \emph{non-negative product polynomial}. Thus, to prove a lower bound on the circuit complexity of $P$, it suffices to lower bound the number of terms in any decomposition as in (\ref{eq:prod-decomp}).

As noted by Jerrum and Snir~\cite{JerrumSnir}, one way to do this is via the \emph{support of the polynomial} $P$, by which we mean the set of monomials that have non-zero coefficients in $P$. We think of this set, denoted $\Supp(P)$, as a subset of $2^{[n]}$ by identifying each multilinear monomial on $x_1,\ldots,x_n$ with a subset of $[n]$ in the natural way. Given a decomposition of $P$ into non-negative product polynomials as in (\ref{eq:prod-decomp}), we immediately get $\Supp(P) = \bigcup_{i\in [s]}\Supp(g_i\cdot h_i).$ And so it suffices to obtain a $P$ such that any such decomposition of $\Supp(P)$ must have large size. 

Such decompositions are closely related to a model of communication complexity known as \emph{Multipartition Communication Complexity}, introduced by \v{D}uris, Hromkovi\v{c}, Jukna, Sauerhoff and Schnitger~\cite{DHJSS} (see also the earlier result of Borodin, Razborov and Smolensky~\cite{BRS}). The multipartition communication complexity of a subset $\mc{S}\subseteq 2^{[n]}$ (or equivalently a Boolean function $f:\{0,1\}^n\rightarrow \{0,1\}$)) is defined as follows. We define a \emph{rectangle} $\mc{R}\subseteq 2^{[n]}$ to be any set of the form $\{A\cup B\ |\ A\in \mc{A}, B\in \mc{B}\}$, where $\mc{A} \subseteq 2^Y$ and $\mc{B}\subseteq 2^Z$ and $(Y,Z)$ is a partition of $[n]$. Further, we say that both the partition and the rectangle $\mc{R}$ are \emph{balanced} if $|Y|,|Z|\geq n/3.$ Finally, the multipartition communication complexity of $\mc{S}$ is defined to be $\lceil\log_2 k\rceil$ where $k$ is the smallest integer such that $\mc{S}$ can be decomposed as the union of $k$ many balanced rectangles. 

To see the connection to algebraic complexity, note that if $P\in \mathbb{R}[x_1,\ldots,x_n]$ has monotone algebraic circuits of size $s$, then (\ref{eq:prod-decomp}) implies that $\Supp(P)$ has multipartition communication complexity at most $\lceil\log_2 s\rceil.$ In particular, linear lower bounds in this model for some explicit $\mc{S}$ implies that any non-negative polynomial $P$ with support exactly $\mc{S}$ cannot be computed by monotone algebraic circuits of subexponential size. 

Polynomial (but sublinear) lower bounds for multipartition communication complexity were implicit in the work of Borodin et al.~\cite{BRS} and were extended to linear (but somewhat non-explicit) lower bounds in the work of \v{D}uris et al.~\cite{DHJSS}. An explicit linear lower bound for this model is implicit in a result of Bova, Capelli, Mengel and Slivovsky~\cite{BCMS}. (See also the related work of Hayes~\cite{Hayes}. Similar constructions are attributed to Wigderson in~\cite{RazYehudayoff} and carried out by Jukna~\cite{Jukna-tropical}.) The hard problem of~\cite{BCMS} is quite easy to describe. Fix a regular expander graph\footnote{Recall that we call a family of $d$-regular graphs $(G_n)_{n\geq 1}$ (with $G_n$ a graph on $n$ vertices) an expander sequence if the second largest (in absolute value) eigenvalue of its adjacency matrix $A$ is at most $d(1-\Omega(1)).$ For the problem defined above, take $G = G_n$ in such a sequence.} $G$ on vertex set $[n]$ with constant degree $d$. The associated hard problem is given by taking $\mc{S}$ to be the set of all vertex covers in $G$. Said differently, we consider the Boolean function $f_G(x_1,\ldots,x_n) = \bigwedge_{\{i,j\}\in E(G)} (x_i \vee x_j)$.

As mentioned above, the communication complexity lower bound on $\mc{S}$ immediately yields a strongly exponential lower bound on the monotone algebraic complexity of some explicitly defined polynomial. Unfortunately, as observed by Yehudayoff~\cite{amir}, this does not yield a separation between $\mathrm{MVNP}$ and $\mathrm{MVP}.$ This is because the above argument implies that \emph{any} polynomial $P_0$ that has support $\mc{S}$ requires monotone algebraic circuits of exponential size. Yehudayoff showed that for any polynomial $P$ in $\mathrm{MVNP},$ there is a polynomial-sized monotone algebraic circuit that computes a polynomial $Q$ with the same support. In particular,  the polynomial $P_0$ cannot be in $\mathrm{MVNP}$ as that would contradict our lower bound above. Thus, to obtain a separation between $\mathrm{MVNP}$ and $\mathrm{MVP}$ along these lines, some new idea is necessary. 

We take our cue from the multipartition communication complexity lower bound above, but modify it suitably to obtain a somewhat different lower bound candidate polynomial $P$. Our proof method for the lower bound, as in~\cite{amir}, is not just based on the support of $P$, but rather on the sizes of the coefficients of $P$. We define a probability distribution $\mu$ on the monomials of $P$ and show that for any non-negative product polynomial $g_ih_i$ in a decomposition as in (\ref{eq:prod-decomp}), a random monomial (chosen according to $\mu$) has much smaller coefficient in the product polynomial than in $P$. As the product polynomials sum to $P$, there must be many of them. This yields the lower bound. 

We explain this in some more detail below.

\paragraph{Detailed outline.} The heart of the multipartition communication complexity lower bound for the function $f_G$ is a more standard lower bound for the \emph{non-deterministic communciation complexity} of the \emph{Disjointness} problem. Here, the non-deterministic communication complexity of a function $f$ (or equivalently, the set system $\mc{S}\subseteq 2^{[n]}$ given by $f^{-1}(1)$) is defined in a similar way to multipartition communication complexity, except that each balanced rectangle $\mc{R}$ is defined over the \emph{same equipartition} $(Y,Z)$ of $[n]$, which we can take to be the sets $[n/2]$ and $[n]\setminus [n/2]$ respectively; and the Disjointness function $D(x)$ is defined by the Boolean predicate $\bigwedge_{i\in [n/2]} (x_i\vee x_{i+n/2}).$\footnote{Strictly speaking, the Disjointness function is $\bigwedge_{i\in [m]} (\neg x_i \vee \neg x_{i+n/2})$ but we keep this definition for simplicity.} 

The lower bound for the Disjointness function is proved by a standard \emph{Fooling set} argument (see, e.g.,~\cite{RaoYehudayoff}). We consider the $2^{n/2}\times 2^{n/2}$ \emph{communication matrix} $M$, where the rows and columns are labelled by Boolean settings to variables indexed by $Y$ and $Z$ respectively and the $(i,j)$th entry of $M$ is the disjointness predicate evaluated on the corresponding input. Further, assume that the rows are ordered using the lexicograhical ordering of $\{0,1\}^Y$, and the columns are ordered according to the \emph{reverse} lexicographic ordering of $\{0,1\}^Z$. This ensures that for any $i\in [2^{n/2}]$, the diagonal entry $M(i,i)$ corresponds to an input of the form $(a,\overline{a})$ where $a\in \{0,1\}^Y$ and $\overline{a}$ is the bitwise complement of $a$. From the definition of the Disjointness function, one can check that each diagonal entry of $M$ is $1$; further, given $i\neq j$, either $M(i,j)$ or $M(j,i)$ is $0$. This implies that any rectangle over $(Y,Z)$ that contains the $i$th diagonal entry cannot contain the $j$th diagonal entry for any $j\neq i$. In particular, the number of rectangles required to cover all the diagonal entries is $2^{n/2},$ implying a linear lower bound on the non-deterministic communication complexity of the Disjointness function.

For the multipartition setting, we can follow the above strategy to prove a lower bound for the function $f_G$ defined above. The intuition is that for any graph $G$, the function $f_G$ contains many copies of the Disjointness function above. In particular, taking any induced matching $M$ of size $m$ in $G$ and setting variables corresponding to vertices $i\not\in V(M)$ to $1$, we get a copy $f_M$ of the Disjointness function on $2m$ bits.  Given any rectangle $\mc{R}$ over the partition $(Y,Z)$, one can similarly prove that $\mc{R}$ cannot contain many (suitably defined) ``diagonal entries''  of the communication matrix of $f_M$, \emph{as long as} $M$ contains many (say $\Omega(m)$) edges from the cut defined by $(Y,Z)$ in $G$.  

But there is a subtle question of how to choose $M$ as above. In the multipartition setting, the partition $(Y,Z)$ is not known ahead of time and furthermore, each rectangle comes with its own underlying partition. This is where the expanding nature of the graph $G$ comes in. Standard facts about expander graphs imply that given any balanced partition $(Y,Z)$ (i.e. $|Y|,|Z|\geq n/3$), a constant fraction of the edges of $G$ lie in the cut defined by $(Y,Z)$. In particular, choosing $M$ \emph{randomly} guarantees that many edges of $M$ lie in the cut with high probability. This leads to a proof of the multipartition commmunication complexity lower bound.

We now describe how this connects to the lower bounds of this paper for monotone algebraic circuits. We will follow a similar strategy, but instead of the $0$s and $1$s of the Boolean predicate, we will analyze the coefficients of the multilinear monomials in $P$ and in the terms of the decomposition in (\ref{eq:prod-decomp}). The polynomial $P$ is defined using an expander graph $G$ on vertex set $[n]$ (let us skip over what the definition of $P$ is for the moment) and the hard distribution $\mu$ over the monomials of $P$ is again just the process of choosing a random induced matching\footnote{For some technical reasons, we will actually choose $M$ so that the non-adjacent vertices of $M$ are at distance at least $3$ from each other. But this can be ignored for now.} $M$ of size $m$ in $G$ and considering the monomial $\prod_{i\in V(M)}x_i$.

The proof of the lower bound then proceeds as follows. Assume that $P$ has a circuit of size $s$ and consider the decomposition given in (\ref{eq:prod-decomp}). Given a term $g_ih_i$ of the decomposition, we get a balanced partition $(Y_i,Z_i)$ of the underlying variable set $x_1,\ldots,x_n.$ We argue that for a random monomial $\mon$ chosen according to the distribution $\mu,$ the expected value of the coefficient of $\mon$ in $g_ih_i$ is much smaller than its coefficient in $P$. To do this, we use a numerical analogue of the fooling set technique outlined above. Again, we consider the ``communication matrix'' $M$, which now is a $2^{|Y_i|}\times 2^{|Z_i|}$ matrix whose rows and  columns are labelled by multilinear monomials in $Y_i$ and $Z_i$ respectively, and such that the entry corresponding to monomials $(\mon_1,\mon_2)$ is the coefficient of the product monomial $\mon_1\cdot \mon_2$ in $P$. The main technical part of the proof shows the following: for independently sampled monomials $\mon'$ and $\mon''$ (chosen from distribution $\mu$) that factor as $\mon_1'\cdot \mon_2'$ and $\mon_1''\cdot \mon_2''$ respectively, where $\mon_1',\mon_1''$ are monomials over $Y_i$ and $\mon_2',\mon_2''$ are monomials over $Z_i$, the coefficients of the ``cross monomials'' $\hat{\mon} := \mon_1'\cdot \mon_2''$ and $\tilde{\mon} := \mon_1''\cdot \mon_2'$ in $P$ are much smaller than the coefficients of $\mon'$ and $\mon''$ in $P$.  This immediately implies that the coefficients of $\mon'$ and $\mon''$ in $g_ih_i$ are smaller than they are in $P$ by the following simple argument. If we let $\coeff(\mon, Q)$ denote the coefficient of monomial $\mon$ in a polynomial $Q$, then we see that
\begin{align*}
\coeff(\mon',g_ih_i)\cdot \coeff(\mon'',g_ih_i) &= \coeff(\mon_1',g_i) \coeff(\mon_2',h_i)\coeff(\mon_1'',g_i)\coeff(\mon_2'',h_i)\\
&= \coeff(\hat{\mon},g_ih_i) \cdot \coeff(\tilde{\mon},g_ih_i).
\end{align*}
The latter term is upper bounded by the product of the coefficients of the monomials $\hat{\mon}$ and $\tilde{\mon}$ in $P$ (because of the decomposition  (\ref{eq:prod-decomp})), which we already argued are much smaller than the coefficients of $\mon'$ and $\mon''$ in $P$. This implies that a randomly chosen monomial $\mon$ has much smaller coefficient in any product term $g_ih_i$ than in $P$. Therefore, there must be many such terms in the decomposition (\ref{eq:prod-decomp}). This implies the lower bound.

The above outline also indicates the property of $P$ that allows the lower bound proof to work: we would like that the coefficients of $\mon'$ and $\mon''$ are much larger than those of the monomials $\hat{\mon}$ and $\tilde{\mon}$. We do this by designing a polynomial $P$ in $\mathrm{MVNP}$ where the coefficients of any monomial $\prod_{i\in S} x_i$ grows with the number of edges in the subgraph of $G$ induced by the set $S$. Recall that for a random monomial $\mon$ chosen according to $\mu$, $S$ is the vertex set of a matching of size $m$ and hence this induced subgraph has $m$ edges. However, if $m$ is sufficiently smaller than $n$ (say $m\leq \alpha n$ for a small enough $\alpha > 0$), we do not expect the vertex sets of two independently chosen matchings of size $m$ to have too many edges between them. This is what allows us to bound the coefficients of $\hat{\mon}$ and $\tilde{\mon}$, and prove the lower bound as above.

\section{Defining the hard polynomial}

\paragraph{Notation.} Throughout, let $n\geq 1$ be a growing integer parameter. Let $X = \{x_1,\ldots,x_n\}$ be a set of indeterminates. We use $x^S$ to denote the monomial $\prod_{i\in S}x_i.$ Given a polynomial $P\in \mathbb{R}[x_1,\ldots,x_n]$ and $S\subseteq [n]$, we use $\coeff(x^S,P)$ to denote the coefficient of the monomial $x^S$ in the polynomial $P$. 

Let $(G_n)_{n > d}$ be an explicit sequence of $d$-regular expander graphs on $n$ vertices with second largest eigenvalue at most $d^{0.75}.$ Here, $d$ is a large enough constant as specified below. Such an explicit sequence of expander graphs can be constructed using, say,~\cite{RVW}. The only fact we will use about expanders is the following, which is an easy consequence of the Expander Mixing Lemma~\cite{AC} (see also~\cite[Lemma 2.5]{HLW}).

For any pair of \emph{disjoint} sets $U,V\subseteq V(G_n)$, we use $E(U,V)$ to denote the set of edges $\{u,v\}\in E(G_n)$ such that $u\in U$ and $v\in V$. Also, let $E(U)$ denote the set of edges $e = \{u,v\}\in E(G_n)$ such that $u,v\in U$.

\begin{lemma}[Corollary to Expander Mixing Lemma]
\label{lem:eml}
Let $G_n$ be as above. Then, for any disjoint sets $U,V\subseteq [n]$ such that $|U|,|V|\in [n/3,2n/3]$, we have 
\[
|E(U,V)|\geq \frac{|E(G_n)|}{10}.
\]
as long as $d$ is a large enough constant. 
\end{lemma}

From now on, $d$ will be fixed to be a large enough constant so that the inequality in Lemma~\ref{lem:eml} holds. 

We define the polynomial $P_n(x_1,\ldots,x_n)$ as follows. We assume that $V(G_n) = [n]$. For each edge $e\in E(G_n)$ introduce a variable $x'_e$ and let $X' = \{x'_e\ |\ e\in E(G_n)\}.$ Notice that for each Boolean assignment to the variables in $X'$, we obtain a subgraph $H$ of $G_n$. In particular, if the variables in $X'$ are set \emph{randomly} to Boolean values, we get a random subgraph $H$ of $G_n$ with the same vertex set $[n]$. We use $\deg_H(i)$ to denote the degree of the vertex $i$ in the graph $H$. 

We now define
\begin{align}
P_n(x_1,\ldots,x_n) &= \avg{\substack{x'_e\in \{0,1\}\\ \forall e\in E(G_n)}}{\prod_{i\in [n]}\left(1+x_i\cdot 2^{\deg_{H}(i)}\right)}\label{eq:defP.1}\\
 &= \avg{\substack{x'_e\in \{0,1\}\\ \forall e\in E(G_n)}}{\sum_{S\subseteq [n]}x^S\cdot 2^{\sum_{i\in S}\deg_{H}(i)}}\label{eq:defP.2}
\end{align} 
where the variables $x'_e$ are set to one of $\{0,1\}$ independently and uniformly at random.

\begin{lemma}
\label{lem:ubd}
The sequence of polynomials $P_n$ as defined above is in $\mathrm{MVNP}$. 
\end{lemma}

\begin{proof}
Using (\ref{eq:defP.1}), we see that 
\[
P_n(x_1,\ldots,x_n)  = \frac{1}{2^{|E(G_n)|}}\sum_{x'_e\in \{0,1\}: e\in E(G)} \prod_{i\in [n]}\left(1+x_i\cdot 2^{\sum_{e\ni i} x'_e}\right).
\]
Since $G_n$ is $d$-regular, it suffices to show that each function $f:\{0,1\}^d\rightarrow \mathbb{R}$ defined by $f(x'_1,\dots,x'_d) = 2^{\sum_{j\in [d]} x'_j}$ can be represented by a constant-sized polynomial over $x'_1,\ldots,x'_d$ with \emph{non-negative} coefficients. 

But this is clear since $f(x_1',\ldots,x_d') = \sum_{S\subseteq [d]}\prod_{i\in S} x_i'.$
\end{proof}

%\begin{remark}
%\label{rem:2toc}
%The above lemma also holds if we change the definition of $P_n$ in (\ref{eq:defP.1}) with the constant $2$ replaced by any fixed $c > 1$.
%\end{remark}

\section{The lower bound}

The main theorem of this section is the following. 

\begin{theorem}
\label{thm:lbd}
Any monotone circuit computing $P_n$ has size $2^{\Omega(n)}.$
\end{theorem}

We need the following lemma from~\cite{RazYehudayoff}. We say that a pair of multilinear polynomials $(g,h)\in \mathbb{R}[X]$ form a \emph{non-negative product pair} if $g,h$ are polynomials with non-negative coefficients, and there is a partition of $X = Y\cup Z$ where $n/3\leq |Y|,|Z|\leq 2n/3$ and $g\in \mathbb{R}[Y], h\in \mathbb{R}[Z].$

\begin{lemma}[\cite{RazYehudayoff}, Lemma 3.3]
\label{lem:amir}
Assume that $P_n$ has a monotone circuit of size $s$. Then 
\[
P_n(X) = \sum_{i=1}^{s+1} g_ih_i
\]
where for each $i\in [s]$, $(g_i,h_i)$ forms a non-negative product pair. 
\end{lemma}

\begin{corollary}
\label{cor:amir}
Assume that $P_n$ has a monotone circuit of size $s$. Let $\mu$ be any probability distribution on subsets $S\subseteq [n]$. Then, there is a non-negative product pair $(g,h)$ such that
\begin{itemize}
\item $gh\leq P$, i.e., $\coeff(x^S,gh)\leq \coeff(x^S,P_n)$ for each $S\subseteq [n]$, 
\item $\avg{S\sim \mu}{\coeff(x^S,gh)/\coeff(x^S,P_n)}\geq 1/(s+1).$ (The  quantity $\coeff(x^S,gh)/\coeff(x^S,P_n)$ is well defined since by (\ref{eq:defP.2}), the denominator is non-zero for all $S\subseteq [n]$.)
\end{itemize}
\end{corollary}

\begin{proof}
Write $P_n = \sum_{i\leq s+1} g_ih_i$ as in Lemma~\ref{lem:amir}. For any fixed $S\subseteq [n]$ and a uniformly random $i\in [s+1]$, we have
\[
\avg{i\in [s+1]}{\frac{\coeff(x^S,g_ih_i)}{\coeff(x^S,P_n)}} = \frac{1}{s+1}\sum_{i\in [s]} \frac{\coeff(x^S,g_ih_i)}{\coeff(x^S,P_n)} = \frac{1}{s+1}.
\]
In particular, the above also holds when $S$ is chosen according to $\mu$. The result now follows by averaging over $i\in [s+1]$. 
\end{proof}

Given Corollary~\ref{cor:amir}, to prove Theorem~\ref{thm:lbd}, it suffices to show the following.

\begin{lemma}
\label{lem:main}
There is a probability distribution $\mu$ on subsets $S\subseteq [n]$ such that for any non-negative product pair $(g,h)$ with $gh\leq P_n$, we have
\begin{equation}
\label{eq:cover}
\avg{S\sim \mu}{\coeff(x^S,gh)/\coeff(x^S,P_n)} \leq \exp(-\Omega(n)).
\end{equation}
\end{lemma}

We need some preparatory work before proving Lemma~\ref{lem:main}.

\begin{lemma}
\label{lem:Palternate}
There exist constants $A, B > 1$ such that 
\[
P_n(X) = \sum_{S\subseteq [n]} x^S B^{|S|} A^{|E(S)|}.
\]
\end{lemma}

\begin{proof}
Using (\ref{eq:defP.2}), we obtain
\begin{align*}
P_n(x_1,\ldots,x_n) &= \avg{x'_e: e\in E(G)}{\sum_{S\subseteq [n]}x^S\cdot 2^{\sum_{i\in S}\deg_{H}(i)}} = \sum_{S\subseteq [n]} x^S\cdot \avg{x'_e: e\in E(G)}{2^{\sum_{i\in S}\deg_{H}(i)}}\\
\end{align*}
where $H$ is the random subgraph of $G$ defined by a uniformly random Boolean assignment to the variables in $X'$. Note that 
\[
\sum_{i\in S}\deg_{H}(i) = \sum_{i\in S}\sum_{e\ni i} x'_e = \sum_{e\in E(S,\bar{S})} x'_e+\sum_{e\in E(S)}2x'_e.
\]
Hence, we get for any $S\subseteq [n]$, 
\begin{align*}
\avg{x'_e: e\in E(G)}{2^{\sum_{i\in S}\deg_{H}(i)}} &= \prod_{e\in E(S,\bar{S})}\avg{x'_e}{2^{x'_e}} \cdot \prod_{e\in E(S)}\avg{x'_e}{4^{x'_e}} = (3/2)^{|E(S,\bar{S})|}\cdot (5/2)^{|E(S)|}\\
&= (3/2)^{|S|d}\cdot \frac{(5/2)^{|E(S)|}}{(3/2)^{2|E(S)|}}
\end{align*}
where for the last equality, we have used the fact that $2|E(S)| + |E(S,\bar{S})| = |S|d.$ Note that this proves the lemma with $B = (3/2)^d$ and $A = (10/9).$
\end{proof}

We now define the probability distribution $\mu$ that will be shown to have the property in (\ref{eq:cover}). The distribution is defined by the following sampling process. Let $m = \alpha n$ where $\alpha\in (0,1)$ is a small constant specified below.

\vspace*{0.5cm}
\noindent
Sampling Algorithm $\mc{S}:$
\begin{enumerate}
\item Set $M = \emptyset.$ (Eventually, $M$ will be a matching of size $m/2$ in $G$.)
\item For $i = 1$ to $(m/2)$, do the following.
\begin{enumerate}
\item Remove all vertices from $G_n$ that are at distance at most $2$ from any vertex in the matching $M$. Let $G_n^{(i)}$ be the resulting graph. 
\item Choose a uniformly random edge $e_i$ from $E(G_n^{(i)})$ and add it to $M$. 
\end{enumerate}
\item Output $M$. 
\end{enumerate}
The above algorithm defines a distribution $\nu$ over matchings $M$ in $G_n$ of size $m/2$. We define $S = V(M)$ to be the set of vertices sampled by the algorithm. This defines a probability distribution $\mu$ over subsets of $[n]$. 

We will need the following properties of the above algorithm.
\begin{lemma}[Properties of $\mc{S}.$]
\label{lem:propsS}
Let $M$ be sampled as in $\mc{S}$ above and let $S = V(M)$. Then we have
\begin{enumerate}
\item Assuming that $\alpha \leq 1/(100\cdot d^2)$, we have $|M| = (m/2)$, $|S| = m$ and $E(S) = M$ with probability $1$. 
\item Let $(U,V)$ be any partition of $V(G_n)$ such that $n/3 \leq |U|,|V|\leq 2n/3.$ Then, as long as $\alpha \leq 1/(100\cdot d^2)$, for some absolute constant $\gamma > 0,$ we have 
\[
\prob{M}{|M\cap E(U,V)|\leq \gamma m}\leq \exp(-\gamma m).
\]
\item Let $M_1$ and $M_2$ be two independent samples obtained by running $\mc{S}$ twice, and let $S_i = V(M_i)$ ($i\in [2]$). Let $(U,V)$ be a partition of $V(G_n)$ as above. Define $R_i = S_i \cap U$ and $T_i = S_i \cap V.$ Then, for $\alpha \leq \gamma\ln A/(100\cdot A^4d^2)$ we have
\[
\avg{M_1,M_2}{A^{|E(R_1,T_2)| + |E(R_2,T_1)|}} \leq A^{\gamma m/4}.
\]
Here, $\gamma$ is as in the previous item and $A$ is as in the statement of Lemma~\ref{lem:Palternate}.
\end{enumerate}
\end{lemma}

\begin{proof}
Item 1 easily follows from the definition of the Sampling algorithm $\mc{S}.$ Note that in each iteration of Step 2, we remove at most $2\cdot (1+d+d^2)$ vertices and hence at most $2(d+d^2+d^3)$ edges from the graph $G_n$. Hence, the upper bound on $\alpha$ guarantees that after $i < (m/2)$ iterations of the for loop, the number of edges removed from the graph is at most 
\[
2i\cdot (d^3+d^2+d) < 4md^3 = 4\alpha nd^3 < \frac{nd}{2},
\]
which allows the algorithm to choose an edge from the graph $G_n^{(i)}$ to add to the  matching $M$.

For Item 2, we proceed as follows. For $i\in \{1,\ldots,m/2\}$, let $e_i$ be the edge chosen by the sampling algorithm $\mc{S}$ in the $i$th iteration of Step 2. Fix any choices of all the $e_j$ with $j < i$ and consider the $i$th iteration of Step 2. The probability that $e_i$ lies in $E(U,V)$ is $|E_i(U,V)|/|E(G_n^{(i)})|$ where $E_i(U,V)$ is the set of edges in $G_n^{(i)}$ with one endpoint each in $U$ and $V$. Note that
\[
|E_i(U,V)|\geq |E(U,V)| - |E(G_n)\setminus E(G_n^{(i)})| \geq \frac{nd}{10} - 2(i-1)\cdot (d^3+d^2+d) \geq \frac{nd}{10}-\alpha n (3d^3) \geq \frac{nd}{20}
\]
where the second inequality follows from Lemma~\ref{lem:eml} and an analysis similar to Item 1 above; the last two inequalities follow from the fact that $2(i-1)< m = \alpha n \leq n/(100\cdot d^2).$ Hence, we have shown that for each $i$,
\[
\prob{}{e_i\in E(U,V)\ |\ e_1,\ldots,e_{i-1}} = \frac{|E_i(U,V)|}{|E(G_n^{(i)})|} \geq \frac{(nd)/20}{(nd)/2} = \frac{1}{10}.
\]
In particular, for any $T\subseteq [m]$, the probability that for every $i\in T$, $e_i\not\in E(U,V)$ can be upper bounded by $(9/10)^{|T|}.$

Thus, the probability that $|M\cap E(U,V)|\leq \ell = \gamma m$ can be bounded by
\begin{align*}
\prob{M}{\exists T\in \binom{[m]}{m-\ell}\ s.t.\ \forall i\in T, e_i\not\in E(U,V) }
&\leq \sum_{T} \prob{M}{\forall i\in T, e_i\not\in E(U,V) }\\
&\leq \binom{m}{\ell} \left(\frac{9}{10}\right)^{m-\ell} \leq \left(\frac{em}{\ell}\right)^\ell\cdot \left(\frac{9}{10}\right)^{m-\ell}\\
&= \left(\frac{e}{\gamma}\cdot \left(\frac{9}{10}\right)^{(1/\gamma)-1}\right)^{\gamma m}\leq \exp(-\gamma m)
\end{align*}
as long as $\gamma$ is bounded by a small enough absolute constant. This finishes the proof of Item 2.

We now prove Item 3. Fix any possible matching $M_1$ as sampled by the algorithm $\mc{S}.$ It suffices to bound $\avg{M_2}{A^{|E(R_2,T_1)| + |E(R_1,T_2)|}}$ for each such $M_1$. Let $\tilde{S}_1$ denote the set of vertices that are at distance at most $1$ from $S_1$ and let $E_1$ denote the set of edges that have at least one endpoint in $\tilde{S}_1$. Note that $|E_1|\leq |\tilde{S}_1|d \leq |S_1|d^2 = md^2 = \alpha nd^2.$

We claim that $|E(R_2,T_1)| + |E(R_1,T_2)|\leq 4|E_1\cap M_2|$. The reason for this is that if a vertex $i\in S_2$ is incident to an edge $e$ in $E(R_1,T_2)\cup E(R_2,T_1)$ then $i\in \tilde{S}_1$ and hence the edge $e'\in M_2$ involving $i$ is an edge in $E_1\cap M_2$. In particular, the number of such vertices $i\in S_2$ is at most $2|E_1\cap M_2|.$ Further, each such vertex $i$ is adjacent to at most $2$ vertices in $S_1$ since vertices in $S_1$ that are not adjacent via an edge in $M_1$ are at distance at least $3$ from each other. Thus each such vertex $i$ contributes at most $2$ to $|E(R_2,T_1)| + |E(R_1,T_2)|.$ This yields the claimed inequality. Thus it suffices to bound $\avg{M_2}{A^{4|E_1\cap M_2|}}.$

We start with a tail bound for $|E_1\cap M_2|.$ 
Let $M_2 = \{e_1',\ldots,e_{m/2}'\}$ where $e_j'$ is the $j$th edge added to $M_2$ by the algorithm $\mc{S}.$ Conditioned on $e_1',\ldots,e_{j-1}',$ the probability that $e_j' \in E_1$ is at most
\[
\frac{|E_1|}{|E(G_n^{(j)})|} = \frac{|E_1|}{|E(G_n)|-|E(G_n)\setminus E(G_n^{(j)})|} \leq \frac{\alpha nd^2}{(nd/2)-3\alpha n d^3} \leq \frac{\alpha nd^2}{(nd)/4} = 4\alpha d
\]
where for the first inequality we have bounded $|E(G_n)\setminus E(G_n^{(j)})|$ and $|E_1|$ as above and for the second inequality we have used the bound on $\alpha.$ Hence, we have
\[
\prob{M_2}{|E_1\cap M_2|\geq i} \leq \sum_{T\in \binom{m/2}{i}}\prob{M_2}{\forall j\in T, e_j'\in E_1}\leq \binom{m/2}{i}(4\alpha d)^i.
\]

This allows us to bound $\avg{M_2}{A^{4|E_1\cap M_2|}}$ for any fixed $M_1$ output by $\mc{S}$.
\begin{align*}
\avg{M_2}{A^{4|E_1\cap M_2|}} &\leq \sum_{i = 0}^{m/2} A^{4i}\prob{M_2}{|E_1\cap M_2|\geq i}\\
&\leq \sum_{i = 0}^{m/2} A^{4i} \cdot \binom{m/2}{i} (4\alpha d)^i = \left(1+4\alpha d A^4\right)^{m/2}\\
&\leq (1+(\gamma \ln A)/2)^{m/2} \leq \exp((m\gamma \ln A )/4) = A^{\gamma m/4}
\end{align*}
where the third inequality follows from the bound $\alpha \leq \gamma\ln A/(100\cdot A^4d^2)$ assumed in the statement of the lemma. 
\end{proof}

We are now ready to prove Lemma~\ref{lem:main}, which will complete the proof of Theorem~\ref{thm:lbd}. 

\begin{proof}[Proof of Lemma~\ref{lem:main}.]
We set $m = \alpha n$ so that $\alpha $ is a positive constant upper bounded by $\gamma \ln A/(100\cdot A^4\cdot d^2)$ and $m$ is even. Assume that $M$ is as sampled above by sampling algorithm $\mc{S}$ and $S = V(M)$. This defines the distribution $\mu$ on subsets of $[n]$. 

Let $(g,h)$ be any non-negative product pair such $gh\leq P_n$. Consequently, there exists a partition $(U,V)$ of $V(G_n) = [n]$ such that $n/3\leq |U|,|V|\leq 2n/3$ and $g\in \mathbb{R}[x_i: i\in U], h\in \mathbb{R}[x_j: j\in V]$. 

Let $\mc{E} = \mc{E}(M)$ denote the event that $|M\cap E(U,V)|\leq \gamma m.$ By Lemma~\ref{lem:propsS} item 2, we know that $\prob{M}{\mc{E}}\leq \exp(-\Omega(n))$ and hence we have 
\begin{align}
\avg{M}{\frac{\coeff(x^S,gh)}{\coeff(x^S,P_n)}} &\leq \avg{M}{\frac{\coeff(x^S,gh)}{\coeff(x^S,P_n)}\  |\ \mc{E}}\prob{M}{\mc{E}} + \avg{M}{\frac{\coeff(x^S,gh)}{\coeff(x^S,P_n)}\ |\ \overline{\mc{E}}}\prob{M}{\bar{\mc{E}}}\notag\\
&\leq \prob{M}{\mc{E}} + \avg{M}{\frac{\coeff(x^S,gh)}{\coeff(x^S,P_n)}\ |\ \overline{\mc{E}}}\notag\\
&\leq  \exp(-\Omega(n)) +\frac{1}{B^m\cdot A^{m/2}} \avg{M}{\coeff(x^S,gh)\ |\ \overline{\mc{E}}}
\label{eq:main.1}
\end{align}
where for the second inequality we have used that $gh\leq P_n$, and for the final inequality we have used our bound on $\prob{M}{\mc{E}}$ along with Lemma~\ref{lem:Palternate} and Lemma~\ref{lem:propsS} item 1.

We now bound the latter term in (\ref{eq:main.1}). For any $i,j,k,$ let $\mc{E}_{i,j,k} = \mc{E}_{i,j,k}(M)$ denote the event that $|M\cap E(U,V)| = i, |M\cap E(U)| = j,$ and $|M\cap E(V)| = k.$ The event $\bar{\mc{E}}$ is partitioned into $\mc{E}_{i,j,k}$ where $i+j+k = (m/2)$ and $i\geq \gamma m$. Let $\mc{T}$ denote the set of such triples $(i,j,k).$ We have

\begin{align*}
\avg{M}{\coeff(x^S,gh)\ |\ \overline{\mc{E}}} &= \sum_{(i,j,k)\in \mc{T}} \avg{M}{\coeff(x^S,gh)\ |\ \mc{E}_{i,j,k}}\cdot \prob{M}{\mc{E}_{i,j,k}\ |\ \overline{\mc{E}}}\\
\end{align*}

Call a triple $(i,j,k)\in \mc{T}$ \emph{heavy} if $\prob{M}{\mc{E}_{i,j,k}}\geq A^{-\gamma m/4}$ and \emph{light} otherwise. Note that as $\prob{M}{\overline{\mc{E}}} = 1-\exp(-\Omega(n)) \geq 1/2,$ we have $\prob{M}{\mc{E}_{i,j,k}\ |\ \overline{\mc{E}}}\leq 2 \prob{M}{\mc{E}_{i,j,k}}.$ In particular, if $(i,j,k)$ is light, we have $\prob{M}{\mc{E}_{i,j,k}\ |\ \overline{\mc{E}}}\leq 2A^{-\gamma m/4}= \exp(-\Omega(n)).$ Plugging this into the expression above, we get

\begin{align}
&\avg{M}{\coeff(x^S,gh)\ |\ \overline{\mc{E}}} = \sum_{(i,j,k)\in \mc{T}} \avg{M}{\coeff(x^S,gh)\ |\ \mc{E}_{i,j,k}}\cdot \prob{M}{\mc{E}_{i,j,k}\ |\ \overline{\mc{E}}}\notag\\
&\leq |\{(i,j,k)\ |\ \text{$(i,j,k)$ light}\}|\cdot B^mA^{m/2}\cdot \exp(-\Omega(n)) + \max_{(i,j,k) \text{ heavy}}\avg{M}{\coeff(x^S,gh)\ |\ \mc{E}_{i,j,k}}\notag\\
&\leq \exp(-\Omega(n))B^mA^{m/2} + \max_{(i,j,k) \text{ heavy}}\avg{M}{\coeff(x^S,gh)\ |\ \mc{E}_{i,j,k}}.\label{eq:main.2}
\end{align}

It suffices therefore to bound $\avg{M}{\coeff(x^S,gh)\ |\ \mc{E}_{i,j,k}}$ for any heavy $(i,j,k)$. This is the main part of the proof. 

Fix some $(i,j,k)\in \mc{T}$ that is heavy. Let $C = \avg{M}{\coeff(x^S,gh)\ |\ \mc{E}_{i,j,k}}.$ Thus, we get
\[
C^2 = \avg{M_1,M_2}{\coeff(x^{S_1},gh)\coeff(x^{S_2},gh)}
\]
where $M_1$ and $M_2$ are independent samples of $M$ conditioned on the event $\mc{E}_{i,j,k}(M),$ and $S_\ell = V(M_\ell)$ for $\ell\in \{1,2\}.$ Define $R_\ell = S_\ell \cap U$ and $T_\ell = S_\ell \cap V.$ We make some simple observations. For each $\ell\in [2]$
\begin{enumerate}
\item $|R_\ell| = i+2j$ and $|T_\ell| = i+2k$,
\item $|E(R_\ell)| = j, |E(T_\ell)| = k$,
\item $\coeff(x^{S_\ell},gh) = \coeff(x^{R_\ell},g)\cdot \coeff(x^{T_\ell},h).$
\end{enumerate}
Thus, we have
\begin{align}
C^2 &= \avg{M_1,M_2}{\coeff(x^{R_1},g)\coeff(x^{T_1},h)\coeff(x^{R_2},g)\coeff(x^{T_2},h)}\notag\\
& = \avg{M_1,M_2}{\coeff(x^{R_1\cup T_2},gh)\coeff(x^{R_2\cup T_1},gh)}\notag\\
&\leq \avg{M_1,M_2}{\coeff(x^{R_1\cup T_2},P_n)\coeff(x^{R_2\cup T_1},P_n)}\notag\\
&= \avg{M_1,M_2}{B^{|R_1|+|T_2|}\cdot A^{|E(R_1\cup T_2)|}\cdot B^{|R_2|+|T_1|}\cdot A^{|E(R_2\cup T_1)|}}\notag\\
&= \avg{M_1,M_2}{B^{4(i+j+k)}\cdot A^{|E(R_1)|+|E(T_1)|+|E(R_2)|+|E(T_2)|+|E(R_1,T_2)|+|E(R_2,T_1)|}}\notag\\
&= B^{2m}A^{m-2i}\cdot \avg{M_1,M_2}{A^{|E(R_1,T_2)|+|E(R_2,T_1)|}}\notag\\
&\leq B^{2m}A^{m(1-2\gamma)}\cdot \avg{M_1,M_2}{A^{|E(R_1,T_2)|+|E(R_2,T_1)|}}\label{eq:main.3}
\end{align}
where we used the observations above for the equalities and for the final inequality, we used the fact that $i\geq \gamma m$ for all $(i,j,k)\in \mc{T}$.

To bound the latter term in (\ref{eq:main.3}), we consider a similar expression where $M_1$ and $M_2$ are replaced by $M_1'$ and $M_2'$ which are independent random outputs of the algorithm $\mc{S}$ (without any conditioning). In this case, by Lemma~\ref{lem:propsS} item 3, we have
\[
\avg{M_1',M_2'}{A^{|E(R_1',T_2')|+|E(R_2',T_1')|}} \leq A^{\gamma m/4},
\]
where $R_\ell',T_\ell'$ are defined analogously for $\ell \in [2]$. Thus, using Bayes's rule we have
\[
\avg{M_1,M_2}{A^{|E(R_1,T_2)|+|E(R_2,T_1)|}} \leq \frac{\avg{M_1',M_2'}{A^{|E(R_1',T_2')|+|E(R_2',T_1')|}}}{\prob{M_1',M_2'}{\mc{E}_{i,j,k}(M_1') \wedge \mc{E}_{i,j,k}(M_2')}} \leq A^{3\gamma m/4}
\]
where the last inequality uses the fact that $(i,j,k)$ is heavy. Plugging the above into (\ref{eq:main.3}), we have
\[
C \leq B^{m}A^{m/2-5\gamma m/8} = B^m A^{m/2} \exp(-\Omega(n)).
\]
As this holds for any heavy $(i,j,k)$, using (\ref{eq:main.2}) and (\ref{eq:main.1}), we obtain the statement of the lemma.
\end{proof}

\paragraph*{Acknowledgements.} The author is grateful to Mrinal Kumar and Amir Yehudayoff for very helpful discussions and encouragement. This work was done during a visit to the ``Lower Bounds Program in Computational Complexity'' program at the Simons Institute for the Theory of Computing; the author is grateful to the organizers of this program and the Simons Institute for their hospitality. The author is also grateful to Igor Sergeev and an anonymous reviewer for pointing out the large body of work on monotone circuit lower bounds carried out by the Russian mathematical community (see~\cite{gashkov-sergeev}). Finally, the author would like to thank the anonymous reviewers (who reviewed this paper for the ACM Transactions on Computation Theory) for their comments and suggestions.

\bibliographystyle{abbrv}
\bibliography{mvpnp-refs}

\begin{thebibliography}{10}

\bibitem{AC}
N.~Alon and F.~Chung.
\newblock Explicit construction of linear sized tolerant networks.
\newblock {\em Discrete Mathematics}, 72(1):15 -- 19, 1988.

\bibitem{BRS}
A.~Borodin, A.~A. Razborov, and R.~Smolensky.
\newblock On lower bounds for read-k-times branching programs.
\newblock {\em Computational Complexity}, 3:1--18, 1993.

\bibitem{BCMS}
S.~Bova, F.~Capelli, S.~Mengel, and F.~Slivovsky.
\newblock Expander {CNFs} have exponential {DNNF} size.
\newblock {\em CoRR}, abs/1411.1995, 2014.

\bibitem{DHJSS}
P.~Duris, J.~Hromkovic, S.~Jukna, M.~Sauerhoff, and G.~Schnitger.
\newblock On multi-partition communication complexity.
\newblock {\em Inf. Comput.}, 194(1):49--75, 2004.

\bibitem{Gashkov}
S.~B. {Gashkov}.
\newblock {On the complexity of monotone computations of polynomials.}
\newblock {\em {Vestn. Mosk. Univ., Ser. I}}, 1987(5):7--13, 1987.

\bibitem{gashkov-sergeev}
S.~B. Gashkov and I.~S. Sergeev.
\newblock A method for deriving lower bounds for the complexity of monotone
  arithmetic circuits computing real polynomials.
\newblock {\em Sbornik. Mathematics}, 203(10), 10 2012.

\bibitem{Hayes}
T.~P. Hayes.
\newblock Separating the k-party communication complexity hierarchy: an
  application of the {Zarankiewicz} problem.
\newblock {\em Discrete Mathematics {\&} Theoretical Computer Science},
  13(4):15--22, 2011.

\bibitem{HLW}
S.~Hoory, N.~Linial, and A.~Wigderson.
\newblock Expander graphs and their applications.
\newblock {\em Bulletin of the American Mathematical Society}, 43(4):439--561,
  2006.

\bibitem{JerrumSnir}
M.~Jerrum and M.~Snir.
\newblock Some exact complexity results for straight-line computations over
  semirings.
\newblock {\em J. {ACM}}, 29(3):874--897, 1982.

\bibitem{Jukna-tropical}
S.~Jukna.
\newblock Lower bounds for tropical circuits and dynamic programs.
\newblock {\em Theory Comput. Syst.}, 57(1):160--194, 2015.

\bibitem{Kasim-Zade}
O.~M. Kasim-Zade.
\newblock The complexity of monotone polynomials.
\newblock In {\em Proceedings of the {A}ll-{U}nion seminar on discrete
  mathematics and its applications ({R}ussian) ({M}oscow, 1984)}, pages
  136--138. Moskov. Gos. Univ., Mekh.-Mat. Fak., Moscow, 1986.

\bibitem{RaoYehudayoff}
A.~Rao and A.~Yehudayoff.
\newblock {\em Communication Complexity: and Applications}.
\newblock Cambridge University Press, 2020.

\bibitem{RazYehudayoff}
R.~Raz and A.~Yehudayoff.
\newblock Multilinear formulas, maximal-partition discrepancy and mixed-sources
  extractors.
\newblock {\em J. Comput. Syst. Sci.}, 77(1):167--190, 2011.

\bibitem{RVW}
O.~Reingold, S.~Vadhan, and A.~Wigderson.
\newblock Entropy waves, the zig-zag graph product, and new constant-degree
  expanders.
\newblock {\em Annals of mathematics}, pages 157--187, 2002.

\bibitem{Ramprasadgit}
R.~Saptharishi.
\newblock A survey of lower bounds in arithmetic circuit complexity.
\newblock {\em Github survey}, 2015.

\bibitem{Schnorr}
C.-P. Schnorr.
\newblock A lower bound on the number of additions in monotone computations.
\newblock {\em Theoret. Comput. Sci.}, 2(3):305--315, 1976.

\bibitem{ShamirSnir}
E.~Shamir and M.~Snir.
\newblock {\em Lower bounds on the number of multiplications and the number of
  additions in monotone computations}.
\newblock IBM Thomas J. Watson Research Division, 1977.

\bibitem{ShpilkaYehudayoff}
A.~Shpilka and A.~Yehudayoff.
\newblock Arithmetic circuits: {A} survey of recent results and open questions.
\newblock {\em Foundations and Trends in Theoretical Computer Science},
  5(3-4):207--388, 2010.

\bibitem{Valiant}
L.~G. Valiant.
\newblock Completeness classes in algebra.
\newblock In {\em Proceedings of the 11h Annual {ACM} Symposium on Theory of
  Computing, April 30 - May 2, 1979, Atlanta, Georgia, {USA}}, pages 249--261,
  1979.

\bibitem{Valiantnegation}
L.~G. Valiant.
\newblock Negation can be exponentially powerful.
\newblock {\em Theor. Comput. Sci.}, 12:303--314, 1980.

\bibitem{amir}
A.~Yehudayoff.
\newblock Separating monotone {VP} and {VNP}.
\newblock In {\em S{TOC}'19---{P}roceedings of the 51st {A}nnual {ACM} {SIGACT}
  {S}ymposium on {T}heory of {C}omputing}, pages 425--429. ACM, New York, 2019.

\end{thebibliography}

\end{document}